\documentclass[12pt, final, onecolumn, notitlepage]{article}
\pdfoutput=1

\usepackage[backend=biber,style=alphabetic,maxbibnames=99,maxalphanames=3,natbib=true]{biblatex}
\usepackage{amsmath, amsfonts, amssymb, amsthm, graphicx, enumitem, titlesec, comment, hyperref}
\allowdisplaybreaks
\titlelabel{\thetitle.\quad}
\usepackage[margin=1in]{geometry}

\newcommand{\vect}[1]{\mathbf{#1}}

\newcommand{\RR}{\mathbb{R}}

\newcommand\norm[1]{\left\lVert#1\right\rVert}
\newcommand\abs[1]{\left\lvert#1\right\rvert}

\newcommand\parens[1]{\left(#1\right)}

\newcommand*{\pr}[2][]{\text{Pr}\ifx\\\left[#1\right]\\\else_{#1}\fi \left[#2\right]}
\newcommand*{\EE}[2][]{\mathbb{E}\ifx\\\left[#1\right]\\\else_{#1}\fi \left[#2\right]}

\newcommand\qu{\text{quad}}

\newtheorem{ctr}{ctr}[section]
\newtheorem{theorem}[ctr]{Theorem}
\newtheorem{lemma}[ctr]{Lemma}

\newtheorem{claim}[ctr]{Claim}
\newtheorem{defin}[ctr]{Definition}
\newtheorem{question}[ctr]{Question}

\theoremstyle{definition}

\newtheorem{remark}[ctr]{Remark}

\begin{document}
	\begin{center}
		{\Large \textbf{Strictly Proper Contract Functions Can Be Arbitrage-Free}}
		
		\vspace{.2cm}
		Eric Neyman, Tim Roughgarden
		
		\today
	\end{center}
	
\begin{abstract}
	We consider mechanisms for truthfully eliciting probabilistic predictions from a group of experts. The standard approach --- using a proper scoring rule to separately reward each expert --- is not robust to collusion: experts may collude to misreport their beliefs in a way that guarantees them a larger total reward no matter the eventual outcome. Chun and Shachter (2011) termed any such collusion ``arbitrage" and asked whether there is any truthful elicitation mechanism that makes arbitrage impossible. We resolve this question positively, exhibiting a class of strictly proper arbitrage-free contract functions. These contract functions have two parts: one ensures that the total reward of a coalition of experts depends only on the average of their reports; the other ensures that changing this average report hurts the experts under at least one outcome.
\end{abstract}

\section{Introduction}
Suppose that some person or entity (the \emph{principal}) wishes to elicit a probabilistic forecast from an expert: for example, a local news organization may want to ask a meteorologist for the probability that it will rain tomorrow. The principal could incentivize the expert with a payment scheme that depends on the expert's report (in our example, the meteorologist's stated probability) and the eventual outcome (whether or not it rains). Such a payment scheme is called a \emph{scoring rule}, and a scoring rule is called \emph{proper} if the expert's optimal strategy for maximizing expected payment is to report their belief. A scoring rule is called \emph{strictly proper} if this is the expert's unique optimal strategy, no matter their belief. The space of all strictly proper scoring rules is very large, but the most well-known and frequently used one is Brier's quadratic scoring rule \cite{brier50}, which gives an expert who reports $p$ a reward of $1 - 2(1 - p)^2$ if the event happens and $1 - 2p^2$ if it does not. This can be thought of as penalizing the expert based on their report's distance to the ``right answer" (see Equation~\ref{eq:qu_norm} in Section~\ref{sec:prelims}). The logarithmic scoring rule \cite{good52}, which gives a reward of $\ln(p)$ if the event happens and $\ln(1 - p)$ if it does not, is also strictly proper and well known.

In many settings, the principal may wish to elicit forecasts from multiple experts, so as to get a better sense of expert opinion and the extent to which there is a consensus. The principal could use the quadratic scoring rule (or any other strictly proper scoring rule) to elicit each expert's forecast. If experts are not allowed to collude, then this strategy is incentive-compatible; however, \cite{french83} observed that experts can collude in a way that increases the sum total profit of all experts, no matter the final outcome. For example, if three experts believe that there is a 40\%, 50\%, and 90\% chance of rain and are rewarded with the quadratic scoring rule, then their total reward is $0.28+0.5+0.98=1.76$ if it rains and $0.68+0.5-0.62=0.56$ if it does not; but if they collude to all report 60\% then their reward is $0.68+0.68+0.68=2.04$ if it rains and $0.28+0.28+0.28=0.84$ if it does not. The experts can agree beforehand to a redistribution of their rewards in such a way that each expert is guaranteed to be better off than if they had not colluded.

Chun and Shachter \cite{cs11} called this phenomenon --- in which experts collude to misreport in a way where their total reward is larger no matter the outcome --- \emph{arbitrage}. They showed that every strictly proper scoring rule admits arbitrage --- indeed, that there is an arbitrage opportunity for any group of experts so long as they do not all agree on the probability of the event. Specifically, a coalition of experts can risklessly make a profit by deviating to report an aggregate of their beliefs (in the case of the quadratic scoring rule, this aggregate is the arithmetic mean). Recent work by Neyman and Roughgarden \cite{nr21} extend this observation to probability distributions over more than two possible outcomes.\\

For many reasons, the expert may wish to make arbitrage impossible. First, the principal may wish to know whether the experts are in agreement: if they are not, for instance, the principal may want to elicit opinions from more experts. If the experts collude to report an aggregate value (as in our example), the principal does not find out whether they originally agreed. Second, even if the expert only seeks to act based on some aggregate of the experts' opinions, their method of aggregation may be different from the one that experts use to collude. For instance, the principal may have a private opinion on the trustworthiness of each expert and wishes to average the experts' opinions with corresponding weights. Collusion among the experts denies the principal this opportunity. Third, a principal may wish to track the accuracy of each individual expert (to figure out which experts to trust more in the future, for instance), and collusion makes this impossible. Fourth, the space of collusion strategies that constitute arbitrage is large. In our example above, any report in $[0.546, 0.637]$ would guarantee a profit; and this does not even mention strategies in which experts report different probabilities. As such, the principal may not even be able to recover basic information about the experts' beliefs from their reports.

As we have discussed, preventing arbitrage is impossible if the principal chooses a strictly proper scoring rule and uses it to reward all of the experts. However, the principal has more freedom than this: they may choose to make each expert's reward depend not only on that expert's report but also \emph{other} experts' reports. Chun and Shachter ask whether there is any mechanism for rewarding experts that makes arbitrage impossible, concluding that this ``seems unlikely" \cite[\S5]{cs11}. Freeman et al. explore this question further, proposing a mechanism that prevents arbitrage if the experts' reports are guaranteed to be in the range $[\epsilon, 1 - \epsilon]$ for some positive $\epsilon$ (though their mechanism may require very large payments if $\epsilon$ is small) \cite{fppw20}. However, they leave open Chun and Shachter's question of whether an incentive-compatible, arbitrage-free reward mechanism exists.

We resolve this question in the affirmative by exhibiting a class of incentive-compatible mechanisms in which arbitrage from collusion is impossible. Our mechanism takes inspiration from Brier's quadratic scoring rule, but modifies it to take into account the aggregate performance of the remaining experts.

\section{Related Work}
Freeman et al. the question of whether strictly proper arbitrage-free mechanisms exist by proving positive results under different relaxations of these constraints \cite{fppw20}. Their main result is a strictly proper arbitrage-free mechanism under the restriction that the range of allowed reports is restricted to $[\epsilon, 1 - \epsilon]$. However, their mechanism necessitates payments that are exponentially large in $1/\epsilon$. Alternatively, these payments can be scaled down, but at the expense of giving essentially zero reward to each expert on the vast majority of the interval of possible reports, thus providing little incentive for truthful reporting. They also exhibit a positive result if the properness criterion is somewhat relaxed to allow for some contract functions that are proper but not strictly proper.

Chen et al. explore the different but related topic of \emph{arbitrage-free wagering mechanisms} \cite{cdpv14}. In a wagering mechanism, each expert wagers a certain amount of money along with their report, and the pool of wagers is redistributed among the experts depending on each expert's report and wager and the eventual outcome. In this setting, they define arbitrage as any opportunity for an individual to risklessly make a profit. That is, an arbitrage opportunity is one in which an expert may unilaterally deviate by submitting a report that guarantees a profit no matter the final outcome. This differs from Chun and Shachter's definition of arbitrage, which is concerned with riskless profit opportunities stemming from collusion between experts.

The most well-known wagering mechanism is the \emph{weighted score wagering mechanism}, which rewards each expert based on their performance compared to other experts according to a strictly proper scoring rule. An expert may risklessly profit from a weighted score wagering mechanism by reporting an aggregate of other experts' reports. This is the same aggregate as the one that a coalition of experts who are rewarded with a strictly proper scoring rule may report in order to risklessly make a profit in our setting. Chen et al. define \emph{no-arbitrage wagering mechanisms}, which modify the reallocation rule of weighted score wagering mechanisms to reward each expert based on their performance relative to the performance of the aggregate of all other experts' reports \cite[\S4.1]{cdpv14}. Our mechanism mechanism and theirs share some of the same spirit, but are different mechanisms that solve different problems.

\section{Preliminaries} \label{sec:prelims}
We consider an event with $n$ possible outcomes. The space of possible probability distributions over these outcomes is $\Delta_n$, the standard simplex in $\RR^n$. We denote the vertices of this simplex by $\delta_j$ for $j \in [n]$ ($\delta_j$ is the vector whose $j$-th coordinate is $1$ and all of whose other coordinates are $0$).

\paragraph{Scoring rules} A scoring rule is any function that takes as input a probability distribution over the $n$ possible outcomes, and the eventual outcome, and outputs a reward. Formally, a scoring rule is any function $s: \Delta_n \times [n] \to \RR$; if an expert reports probability distribution $\vect{p} \in \Delta_n$ and the outcome is $j$, then the expert receives reward $s(\vect{p}; j)$.

A scoring rule is \emph{proper} if an expert with belief $\vect{b}$ maximizes their expected reward by reporting $\vect{b}$. Formally, $s$ is proper if for all $\vect{b}$, $\sum_j b_j s(\vect{x}; j)$ is maximized at $\vect{x} = \vect{b}$. We say that $s$ is \emph{strictly proper} if for all $\vect{b}$, $\vect{x} = \vect{b}$ is the unique maximizer, i.e. that an expert does strictly worse by misreporting their belief.

\emph{Brier's quadratic scoring rule} --- mentioned in the introduction --- is the scoring rule
\[s_\qu(\vect{p}; j) := 1 - (1 - p_j)^2 - \sum_{\ell \neq j} p_\ell^2.\]
It can be rewritten as
\begin{equation} \label{eq:qu_norm}
	s_\qu(\vect{p}; j) = 1 - \norm{\vect{p} - \delta_j}_2^2.
\end{equation}
Thus, the quadratic scoring rule can be thought of as penalizing the expert by the squared distance between their report and the ``omniscient" answer $\delta_j$. The quadratic scoring rule is strictly proper. In the case of $n = 2$ outcomes, the quadratic scoring rule can be written\footnote{The quadratic scoring rule for two outcomes is more frequently written as $1 - (1 - p_j)^2$, but we choose to include a factor of $2$ to be consistent with the usual formula for $n > 2$ outcomes. The usual form is the same as ours up to a positive affine transformation; such transformations preserve strict properness.} as
\[s_\qu(\vect{p}; j) = 1 - 2(1 - p_j)^2.\]

\paragraph{Contract functions}
Contract functions, defined by Chun and Shachter \cite{cs11}, generalize scoring rules to multiple experts. We say that there are $m$ experts; for $i \in [m]$, expert $i$ reports a probability distribution $\vect{p}_i \in \Delta_n$. We denote the $j$-th coordinate of $\vect{p}_i$ as $p_{i, j}$.

A \emph{contract function} is any function that takes as input the $m$ experts' reports and the outcome, and outputs the reward of each expert. Formally, a contract function is any function $\Pi: (\Delta_n)^m \times [n] \to \RR^m$; if the experts report distributions $\vect{p}_1, \dots, \vect{p}_m$ and the outcome is $j$, then the vector of expert rewards is $\Pi(\vect{p}_1, \dots, \vect{p}_m; j)$. We let $\Pi_i(\cdot)$ denote the $i$-th coordinate of $\Pi(\cdot)$, i.e. expert $i$'s reward. We will generally use $\vect{P}$ to denote the $m$-tuple of reports $(\vect{p}_1, \dots, \vect{p}_m)$.

A contract function is \emph{proper} if for each $i \in [m]$, expert $i$ maximizes their expected reward by reporting their belief $\vect{b}_i$, no matter the reports $\vect{p}_{-i}$ of the other experts. Formally, $\Pi$ is proper if for all $i \in [m]$, for all $\vect{b}_i$ and all $\vect{p}_{-i}$, $\sum_j b_{i, j} \Pi_i(\vect{x}, \vect{p}_{-i}; j)$ is maximized at $\vect{x} = \vect{b}_i$. We say that $\Pi$ is \emph{strictly proper} if $\vect{x} = \vect{b}_i$ is the unique maximizer, i.e. that an expert does strictly worse by misreporting their belief.

Our goal is to exhibit a strictly proper contract function that does not permit arbitrage from collusion. We use the definition of arbitrage given by Freeman et al. \cite{fppw20}, which was adapted from Chun and Shachter \cite{cs11}.

A contract function $\Pi$ \emph{admits arbitrage} if there is a coalition (i.e. subset) $C \subseteq [m]$ of experts and $m$-tuples of expert reports $\vect{P}$ and $\vect{Q}$, with $\vect{p}_i = \vect{q}_i$ for all $i \not \in C$, such that
\[\sum_{i \in C} \Pi_i(\vect{Q}; j) \ge \sum_{i \in C} \Pi_i(\vect{P}; j)\]
for all $j \in [n]$, and the inequality is strict for some $j$. We say that $\Pi$ is \emph{arbitrage-free} if it does not admit arbitrage. Intuitively, $\Pi$ admits arbitrage if it is possible for a coalition of experts to collude to misreport their values in such a way that the total reward of the experts in the coalition ends up larger, no matter the outcome. (Above, the misreport is $\vect{Q}$; the constraint that $\vect{p}_i = \vect{q}_i$ for $i \not \in C$ means that only experts in $C$ change their reports.) If this is possible, then the experts in $C$ can commit beforehand to a redistribution of the extra reward in a way that makes every expert in the coalition better off no matter the eventual outcome $j$.

\begin{remark}
	Positive affine transformations preserve both strict properness and arbitrage-freeness. That is, if $\Pi$ is strictly proper then so is $a\Pi + b$ for any $a > 0$ and $b$, and this is likewise true for arbitrage-freeness.
\end{remark}

The question posed by Chun and Shacther \cite{cs11} and explored by Freeman et al. \cite{fppw20}, which we answer affirmatively in this work, is: \textbf{Does there exist a strictly proper arbitrage-free contract function?}\\

In the case of $m = 2$ experts, there is a fairly straightforward solution:
{ \begin{equation} \label{eq:solution_m2}
		\Pi(\vect{p}_1, \vect{p}_2; j) = \parens{s_\qu(\vect{p}_1; j) - s_\qu(\vect{p}_2; j), s_\qu(\vect{p}_2; j) - s_\qu(\vect{p}_1; j)}.
\end{equation}}%
This contract function is strictly proper because expert 1's reward is the (strictly proper) quadratic score of their report plus a term that does not depend on their report, and likewise for expert 2. It is arbitrage-free because the total reward of the two experts is $0$ no matter what. Indeed, this contract function is arbitrage-free with any strictly proper scoring rule in place of the quadratic scoring rule.

This idea does not extend to $m > 2$ experts, because an arbitrage-free contract function must not admit arbitrage by a coalition of experts of any size. While it is easy to construct a contract function that does not admit arbitrage by a coalition of size $m$ (by making the total reward always equal to $0$), this does not automatically make the contract function free of arbitrage opportunities for coalitions of sizes between $2$ and $m - 1$. In the next section we address this challenge and exhibit a strictly proper contract function that is arbitrage-free for $m > 2$ experts.

\section{A Class of Strictly Proper Arbitrage-Free Contract Functions} \label{sec:mechanism}
Suppose that --- as before --- there are $m \ge 2$ experts who are forecasting an event with $n \ge 2$ outcomes. Given experts with reports $\vect{P} = (\vect{p}_1, \dots, \vect{p}_m)$ and a nonempty subset $S \subseteq [m]$ of the experts, we will let $\overline{\vect{p}}_S := \frac{1}{\abs{S}} \sum_{i \in S} \vect{p}_i$ be the average of the experts' reports. We will use $\overline{\vect{p}}_{-i}$ to denote $\overline{\vect{p}}_{[m] \setminus \{i\}}$.

We now state our main theorem, which exhibits a class of strictly proper, arbitrage-free contract functions.

\begin{theorem} \label{thm:construction}
	Let $\alpha$ be a real number such that $\alpha < 0$ or $\alpha \ge 2(m - 1)^2n$. Let $\Pi$ be the contract function defined by
	\[\Pi_i(\vect{P}; j) = s_\qu(\vect{p}_i; j) - (m - 1)^2 s_\qu(\overline{\vect{p}}_{-i}; j) + \alpha \overline{\vect{p}}_{-i, j}\]
	for each $i$, $j$. Then $\Pi$ is strictly proper and arbitrage-free.
\end{theorem}

Note that in the case of $m = 2$, setting $\alpha = 0$ yields our aforementioned solution for two experts in Equation~\ref{eq:solution_m2}. Unfortunately, setting $\alpha = 0$ for $m > 2$ experts causes arbitrage-freeness to fail in certain edge cases.\\

One can think of the contract function in Theorem~\ref{thm:construction} as having two parts. The first part, $s_\qu(\vect{p}_i; j) - (m - 1)^2 s_\qu(\overline{\vect{p}}_{-i}; j)$, ensures that any coalition's total reward depends only on the average of the coalition's reports. In effect this significantly limits the degrees of freedom that a coalition has when colluding. The second part, $\alpha \overline{\vect{p}}_{-i, j}$, ensures that any deviation in this average report causes a decrease in total reward under at least one outcome.

In this section we focus on proving Theorem~\ref{thm:construction} for $n = 2$ outcomes, as this allows us to simplify notation without sacrificing the core ideas. We defer the proof for general values of $n$ to the appendix.

\begin{proof}[Proof of Theorem~\ref{thm:construction} for $n = 2$]
	First, note that $\Pi$ is strictly proper, because expert $i$'s reward is their quadratic score plus a term that does not depend on their report. It remains to show that $\Pi$ is arbitrage-free.
	
	Let $C \subseteq [m]$ be a coalition of experts. Strict properness entails that no expert can unilaterally find an arbitrage opportunity, so we may assume that $\abs{C} \ge 2$.\\
	
	For an outcome $j$ and a subset $S \subseteq [m]$, let $p_{S, j} := \sum_{i \in S} p_{i, j}$. The following fact follows from algebraic manipulations, which we defer to the appendix.
	\begin{lemma} \label{lem:pmx}
		The expression for $\Pi_i(\vect{P}; j)$ is equal to
		\begin{equation} \label{eq:adapted_n2}
			2(p_{[m], j} - d - 1)(p_{[m], j} - 2p_{i, j} - d + 1) + f(m, \alpha),
		\end{equation}
		for some function $f$, where $d = m - 1 - \frac{\alpha}{4(m - 1)}$.
	\end{lemma}
	
	Equation~\ref{eq:adapted_n2} makes it evident that rewards add nicely across experts in a coalition $C$, as the first term of the product is the same for all experts in $C$. We will use the notation $\Pi_C(\vect{P}; j)$ to denote $\sum_{i \in C} \Pi_i(\vect{P}; j)$. The key idea is that, as we are about to show, if the reports of experts not in $C$ are held fixed, $\Pi_C(\vect{P}; j)$ depends \emph{only} on $p_{C, j}$. Thus, the experts in $C$ have only one degree of freedom available for colluding: the sum of their reports.
	
	We write $\overline{C}$ to mean $[m] \setminus C$. We have
	{
		\begin{align} \label{eq:pi_c_cong}
			\Pi_C(\vect{P}; j) &= 2 \sum_{i \in C} (p_{[m], j} - d - 1)(p_{[m], j} - 2p_{i, j} - d + 1) + \abs{C} f(m, \alpha) \nonumber\\
			&= 2(p_{C, j} + p_{\overline{C}, j} - d - 1) ((\abs{C} - 2)p_{C, j} + \abs{C}(p_{\overline{C}, j} - d + 1)) + \abs{C} f(m, \alpha) \nonumber\\
			&= 2((\abs{C} - 2) p_{C, j}^2 + 2((\abs{C} - 1)(p_{\overline{C}, j} - d) + 1)p_{C, j}) + g(m, \alpha, \abs{C}, p_{\overline{C}, j}),
		\end{align}
	}%
	for some function $g$. Now, recall the constraints on $\alpha$ in Theorem~\ref{thm:construction}, and note that $\alpha < 0 \Leftrightarrow d > m - 1$ and $\alpha \ge 4(m - 1)^2 \Leftrightarrow d \le 0$. With this in mind, we now prove the following claim, which is sufficient to complete our proof.
	
	\begin{claim} \label{claim:monotonicity}
		If $d \le 0$, then for each $j$ and for all possible reports of experts not in $C$, $\Pi_C(\vect{P}; j)$ is a strictly increasing function of $p_{C, j}$. If $d > m - 1$, it is a strictly decreasing function of $p_{C, j}$.
	\end{claim}
	
	By virtue of deriving Equation~\ref{eq:pi_c_cong}, we have already proven the most difficult part of Claim~\ref{claim:monotonicity}, which is that $\Pi_C(\vect{P}; j)$ is a function of (i.e. determined by) $p_{C, j}$. Why is this function's monotonicity sufficient to complete our proof of Theorem~\ref{thm:construction}? Since $p_{C, 1} + p_{C, 2} = \abs{C}$, it follows from Claim~\ref{claim:monotonicity} that for $d \le 0$ and $d > m - 1$, colluding in a way that increases the total reward in the case of one outcome necessarily decreases it in the case of the other outcome.
	
	\begin{proof}[Proof of Claim~\ref{claim:monotonicity}]
		We first consider the case of $\abs{C} = 2$. In this case we have
		\[\Pi_C(\vect{P}; j) = 4(p_{\overline{C}, j} - d + 1)p_{C, j} + g(m, \alpha, \abs{C}, p_{\overline{C}, j}).\]
		Now, $0 \le p_{\overline{C}, j} \le m - 2$, which means that $1 - d \le p_{\overline{C}, j} - d + 1 \le m - 1 - d$. If $d \le 0$, this quantity is guaranteed to be strictly positive, so $\Pi_C(\vect{P}; j)$ is a strictly increasing function of $p_{C, j}$; if $d > m - 1$, it is guaranteed to be negative, so $\Pi_C(\vect{P}; j)$ is a strictly decreasing function of $p_{C, j}$.\\
		
		Now assume that $\abs{C} > 2$. In this case, it follows from Equation~\ref{eq:pi_c_cong} that $\Pi_C(\vect{P}; j)$ is a parabola with a minimum at
		\[\frac{(\abs{C} - 1)(d - p_{\overline{C}, j}) - 1}{\abs{C} - 2}.\]
		We wish to show that if $d \le 0$ then this quantity is at most $0$, and that if $d > m - 1$ then it is at least $\abs{C}$ (since the range of possible values of $p_{C, j}$ is $[0, \abs{C}]$). If $d \le 0$ then, since $p_{\overline{C}, j} \ge 0$, we have
		\[\frac{(\abs{C} - 1)(d - p_{\overline{C}, j}) - 1}{\abs{C} - 2} \le \frac{-1}{\abs{C} - 2} \le \frac{-1}{m - 2} \le 0.\]
		If $d > m - 1$ then, since $p_{\overline{C}, j} \le m - \abs{C}$, we have
		\[\frac{(\abs{C} - 1)(d - p_{\overline{C}, j}) - 1}{\abs{C} - 2} \ge \frac{(\abs{C} - 1)^2 - 1}{\abs{C} - 2} = \abs{C}.\]
	\end{proof}
	
	\noindent Having proved the claim, we have completed the proof of Theorem~\ref{thm:construction} for $n = 2$.
\end{proof}

We note that setting $\alpha = 0$ results in a contract function that is arbitrage-free except in one edge case: in the event that all but two experts assign a probability of zero to some outcome $j$, the remaining experts can collude to adjust their probabilities --- in particular, lowering the total probability they assign to outcome $j$ --- in a way that increases their total reward under outcome $j$ and leaves the remaining rewards unchanged. If we are willing to put this exception aside (e.g. if we only allow reports strictly between $0$ and $1$), then we may regard the resulting contract function $\Pi_i(\vect{P}; j) = s_\qu(\vect{p}_i; j) - (m - 1)^2 s_\qu(\overline{\vect{p}}_{-i}; j)$ as arbitrage-free. This contract function has a natural interpretation: it rewards an expert for the accuracy of their forecast but penalizes the expert if others are accurate in aggregate. This rule is reminiscent of the no-arbitrage wagering mechanism for the quadratic scoring rule given in \cite{cdpv14}, except that the penalty is multiplied by a factor of $(m - 1)^2$.

\section{Future Directions in Arbitrage-Freeness} \label{sec:stronger_defs}
Having exhibited a strictly proper arbitrage-free contract function, it is natural to ask about other (possibly stronger) notions of arbitrage-freeness. Another natural notion is to say that $\Pi$ admits arbitrage if a coalition $C$ of experts can collude in a way that, in the opinion of every expert in $C$, increases the expected total reward of the experts in $C$. Formally:

\begin{defin} \label{def:exp_arb}
	A contract function $\Pi$ \emph{admits expected arbitrage} if there is a coalition $C \subseteq [m]$ of experts and vectors of reports $\vect{P} = (\vect{p}_1, \dots, \vect{p}_m)$, $\vect{Q} = (\vect{q}_1, \dots, \vect{q}_m)$, with $p_i = q_i$ if $i \not \in C$, such that for all $i \in C$ we have
	\[\sum_{j \in [n]} p_{i, j} \sum_{k \in C} \Pi_k(\vect{P}; j) \le \sum_{j \in [n]} p_{i, j} \sum_{k \in C} \Pi_k(\vect{Q}; j),\]
	and the inequality is strict for some $i$. We say that $\Pi$ is \emph{free of expected arbitrage} if it does not admit expected arbitrage.
\end{defin}

Up to edge scenarios,\footnote{It is possible for a coalition of experts to collude in a way that increases their total reward under an outcome to which they all assign probability $0$. If their reward in the case of all other outcome is unchanged, such a deviation would constitute arbitrage but not expected arbitrage.} if a contract function admits arbitrage then it also admits expected arbitrage. On the other hand, in the case of $m > 2$ experts, the scoring rules described by Theorem~\ref{thm:construction} (which do not admit arbitrage) do admit expected arbitrage. As an example, consider two outcomes and $m$ experts with beliefs $(\frac{1}{2}, \frac{1}{2})$. If all experts report their beliefs, then each expert's reward is $\frac{\alpha}{2} + \frac{1}{2}(1 - (m - 1)^2)$ no matter the outcome. If all experts instead report $(1, 0)$ then each expert expects a reward of $\frac{\alpha}{2}$, which is larger. This raises the following question.

\begin{question}
	Is there a strictly proper scoring rule that does not admit expected arbitrage?
\end{question}

We hope that our work will spur research on stronger notions of arbitrage-freeness, included but not limited to Definition~\ref{def:exp_arb}.

\printbibliography

\appendix
\section{Details Omitted from Section~\ref{sec:mechanism}}
\begin{proof}[Proof of Lemma~\ref{lem:pmx}]
	Let $d = m - 1 - \frac{\alpha}{4(m - 1)}$. We use $f_0$ and $f$ to denote functions whose particular values do not concern us. We have
	{
		\begin{align*}
			\Pi_i(\vect{P}; j) &= s_\qu(\vect{p}_i; j) - (m - 1)^2 s_\qu(\overline{\vect{p}}_{-i}; j) + \alpha \overline{\vect{p}}_{-i,j}\\
			&= 1 - 2(1 - p_{i, j})^2 - (m - 1)^2 \parens{1 - 2\parens{1 - \frac{p_{-i, j}}{m - 1}}^2} + \frac{\alpha}{m - 1} p_{-i, j}\\
			&= -2(1 - p_{i, j})^2 + 2(m - 1 - p_{-i, j})^2 + \frac{\alpha}{m - 1}p_{-i, j} + f_0(m, \alpha)\\
			&= -2(1 - p_{i, j})^2 + 2 \parens{d - p_{-i, j}}^2 + f(m, \alpha)\\
			&= 2 \parens{d - p_{-i, j} + (1 - p_{i, j})} \parens{d - p_{-i, j} - (1 - p_{i, j})} + f(m, \alpha)\\
			&= 2(p_{[m], j} - d - 1)(p_{[m], j} - 2p_{i, j} - d + 1) + f(m, \alpha),
	\end{align*}}%
	as desired.
\end{proof}

\begin{proof}[Proof of Theorem~\ref{thm:construction} for general $n$]
	First, note that $\Pi$ is strictly proper, because expert $i$'s reward is their quadratic score plus a term that does not depend on their report. It remains to show that $\Pi$ is arbitrage-free.
	
	Let $C \subseteq [m]$ be a coalition of experts. Strict properness entails that no expert can unilaterally find an arbitrage opportunity, so we may assume that $\abs{C} \ge 2$.
	
	For an outcome $j$ and a subset $S \subseteq [m]$, let $p_{S, j} := \sum_{i \in S} p_{i, j}$. Let $d = m - 1 - \frac{\alpha}{2(m - 1)}$. We rewrite $\Pi_i(\vect{P}; j)$ as follows, using $f_0$ and $f$ to denote functions whose particular values do not concern us.
	
	{
		\begin{align*}
			&\Pi_i(\vect{P}; j) = s_\qu(\vect{p}_i; j) - (m - 1)^2 s_\qu(\overline{\vect{p}}_{-i}; j) + \alpha \overline{\vect{p}}_{-i,j}\\
			&= (m - 1)^2 \parens{1 - \frac{1}{m - 1}p_{-i, j}}^2 - (1 - p_{i, j})^2 + \sum_{\ell \neq j} \parens{(m - 1)^2 \parens{\frac{1}{m - 1} p_{-i, \ell}}^2 - p_{i, \ell}^2}\\
			&\qquad + \frac{\alpha}{m - 1} p_{-i, j} + f_0(m, n, \alpha)\\
			&= (d - p_{-i, j})^2 - (1 - p_{i, j})^2 + \sum_{\ell \neq j} \parens{p_{-i, \ell}^2 - p_{i, \ell}^2} + f(m, n, \alpha)\\
			&= (d - p_{-i, j} + (1 - p_{i, j}))(d - p_{-i, j} - (1 - p_{i, j})) + \sum_{\ell \neq j} (p_{-i, \ell} + p_{i, \ell})(p_{-i, \ell} - p_{i, \ell}) + f(m, n, \alpha)\\
			&= (p_{[m], j} - d - 1)(p_{[m], j} - 2p_{i, j} - d + 1) + \sum_{\ell \neq j} p_{[m], \ell}(p_{[m], \ell} - 2p_{i, \ell}) + f(m, n, \alpha).
	\end{align*}}%
	
	We will use the notation $\Pi_C(\vect{P}; j)$ to denote $\sum_{i \in C} \Pi_i(\vect{P}; j)$. We also write $\overline{C}$ to mean $[m] \setminus C$ and $\vect{P}_{\overline{C}}$ to mean the collection of reports $\vect{p}_i$ for $i \in \overline{C}$. We have
	
	{
		\begin{align*}
			&\Pi_C(\vect{P}; j) = \sum_{i \in C} \parens{(p_{[m], j} - d - 1)(p_{[m], j} - 2p_{i, j} - d + 1) + \sum_{\ell \neq j} p_{[m], \ell}(p_{[m], \ell} - 2p_{i, \ell})} + \abs{C} f(m, n, \alpha)\\
			&= (p_{C, j} + p_{\overline{C}, j} - d - 1)((\abs{C} - 2)p_{C, j} + \abs{C}(p_{\overline{C}, j} - d + 1))\\
			&\qquad + \sum_{\ell \neq j} (p_{C, \ell} + p_{\overline{C}, \ell})(\abs{C} p_{\overline{C}, \ell} + (\abs{C} - 2) p_{C, \ell}) + \abs{C} f(m, n, \alpha)\\
			&= (\abs{C} - 2) p_{C, j}^2 + ((2\abs{C} - 2)(p_{\overline{C}, j} - d) + 2)p_{C, j}\\
			&\quad + \sum_{\ell \neq j} \parens{(\abs{C} - 2) p_{C, \ell}^2 + (2\abs{C} - 2) p_{\overline{C}, \ell} p_{C, \ell}} + g(m, n, \alpha, \abs{C}, \vect{P}_{\overline{C}})\\
			&= (2 - (2\abs{C} - 2)d) p_{C, j} + \sum_\ell \parens{(\abs{C} - 2) p_{C, \ell}^2 + (2\abs{C} - 2) p_{\overline{C}, \ell} p_{C, \ell}} + g(m, n, \alpha, \abs{C}, \vect{P}_{\overline{C}})
	\end{align*}}%
	
	for some function $g$. Consider a different vector $\vect{Q}$ that agrees with $\vect{P}$ on $\overline{C}$.\\
	
	\textbf{Case 1:} $\alpha < 0$. In this case, $2 - (2\abs{C} - 2)d < 2 - (2\abs{C} - 2)(m - 1)$. Let $\tilde{j} = \arg \max_\ell q_{C, \ell} - p_{C, \ell}$, and let $\epsilon = q_{C, \tilde{j}} - p_{C, \tilde{j}}$. We note that
	\[\sum_\ell q_{C, \ell}^2 - p_{C, \ell}^2 = \sum_\ell (q_{C, \ell} - p_{C, \ell})(q_{C, \ell} + p_{C, \ell}) \le \epsilon \sum_\ell (q_{C, \ell} + p_{C, \ell}) = 2 \epsilon \abs{C}.\]
	Thus, we have
	{ \begin{align*}
			\Pi_C(\vect{Q}; \tilde{j}) - \Pi_C(\vect{P}; \tilde{j}) &= (2 - (2\abs{C} - 2)d) \epsilon + \sum_\ell (\abs{C} - 2)(q_{C, \ell}^2 - p_{C, \ell}^2) + (2\abs{C} - 2)(q_{C, \ell} - p_{C, \ell}) p_{\overline{C}, \ell}\\
			&\le (2 - (2\abs{C} - 2)d) \epsilon + (\abs{C} - 2) \cdot 2 \epsilon \abs{C} + (2\abs{C} - 2)(m - \abs{C}) \epsilon\\
			&\le (2 - (2\abs{C} - 2)(m - 1)) \epsilon + (\abs{C} - 2) \cdot 2 \epsilon \abs{C} + (2\abs{C} - 2)(m - \abs{C}) \epsilon\\
			&= 2\epsilon(1 + (\abs{C} - 2)\abs{C} + (\abs{C} - 1)(1 - \abs{C})) = 0,
	\end{align*}}%
	with equality in the second step only when $\epsilon = 0$, i.e. $q_{C, \ell} = p_{C, \ell}$ for all $\ell$. Thus, either the total reward of the experts in $C$ is the same under $\vect{Q}$ as under $\vect{P}$ for every outcome, or it is strictly smaller under $\vect{Q}$ in the case of outcome $\tilde{j}$.\\
	
	\textbf{Case 2:} $\alpha \ge 2(m - 1)^2 n$. In this case, $2 - (2\abs{C} - 2)d \ge 2 + (2\abs{C} - 2)(m - 1)(n - 1)$. $\tilde{j} = \arg \max_\ell p_{C, \ell} - q_{C, \ell}$, and let $\epsilon = p_{C, \tilde{j}} - q_{C, \tilde{j}}$. Since $\sum_\ell (q_{C, \ell} - p_{C, \ell}) = 0$, it follows that $q_{C, \ell} - p_{C, \ell} \le (n - 1)\epsilon$ for all $\ell$. We note that
	\begin{align*}
		\sum_\ell q_{C, \ell}^2 - p_{C, \ell}^2 &\le \sum_\ell (q_{C, \ell} + p_{C, \ell}) \max(q_{C, \ell} - p_{C, \ell}, 0)\\
		&\le 2\abs{C} \sum_\ell \max(q_{C, \ell} - p_{C, \ell}, 0) \le 2 \abs{C}(n - 1) \epsilon.
	\end{align*}
	We also have that
	\[\sum_\ell (q_{C, \ell} - p_{C, \ell}) p_{\overline{C}, \ell} \le (m - \abs{C}) \sum_\ell \max(q_{C, \ell} - p_{C, \ell}, 0)\le (m - \abs{C})(n - 1)\epsilon.\]
	Therefore,
	\begin{align*}
		&\Pi_C(\vect{Q}; \tilde{j}) - \Pi_C(\vect{P}; \tilde{j}) = -(2 - (2\abs{C} - 2)d) \epsilon + \sum_\ell (\abs{C} - 2)(q_{C, \ell}^2 - p_{C, \ell}^2) + (2\abs{C} - 2)(q_{C, \ell} - p_{C, \ell}) p_{\overline{C}, \ell}\\
		&\le -(2 - (2\abs{C} - 2)d) \epsilon + (\abs{C} - 2) \cdot 2\abs{C}(n - 1)\epsilon + (2\abs{C} - 2) \cdot (m - \abs{C})(n - 1)\epsilon\\
		&= 2\epsilon(-1 - (\abs{C} - 1)(m - 1)(n - 1) + (n - 1)(m(\abs{C} - 1) - \abs{C}))\\
		&= 2\epsilon(-1 - (n - 1)) = -2\epsilon n \le 0,
	\end{align*}
	with equality in the last step only when $\epsilon = 0$, i.e. $q_{C, \ell} = p_{C, \ell}$ for all $\ell$. As in the previous case, this means that either the total reward of the experts in $C$ is the same under $\vect{Q}$ as under $\vect{P}$ for every outcome, or it is strictly smaller under $\vect{Q}$ in the case of outcome $\tilde{j}$. This completes the proof.
\end{proof}
\end{document}